\documentclass[11pt,a4paper]{amsart}
\usepackage{tikz,graphicx,amsmath,amssymb,amsthm,amsfonts,enumerate,url,fancyhdr,mathbbol}
\usepackage[pdftex]{hyperref}

\newtheorem{theo}{Theorem}

\newtheorem{lemma}[theo]{Lemma}

\newtheorem{defi}[theo]{Definition}

\theoremstyle{definition}
\numberwithin{theo}{section}

\newtheorem{rem}[theo]{Remark}
\author{Delio Mugnolo}

\title{A Frucht's theorem for quantum graphs}

\address{Institut f\"ur Analysis, Universit\"at Ulm, Helmholtzstra{\ss}e 18, D-89081 Ulm, Germany}
\email{delio.mugnolo@uni-ulm.de}

\begin{document}

\maketitle

\begin{abstract}
A celebrated theorem due to R. Frucht states that, roughly speaking, each group is isomorphic to the symmetry group of some graph. By ``symmetry group'' the group of all graph automorphisms is meant. We provide an analogue of this result for \emph{quantum} graphs, i.e., for Schr\"odinger equations on a metric graph, after suitably defining the notion of symmetry.
\end{abstract}

\section{Introduction}

Beginning with the second half of the XIX century, symmetries have played an important r\^ole in analysis. A key observation that paved the road to the seminal work of Lie, Klein and Noether is the group structure typical of symmetries. Ever since, the notion of symmetry has been very important in the theory of differential equations, but has also appeared in further contexts eventually leading to the typical general question:

\medskip
\emph{Let $\Gamma$ be a group and $C$ be a category. Is there an object in $C$ whose group of symmetries is isomorphic to $\Gamma$?}

Possibly, the earliest example of mathematical problem related to the above question is the following, concerning the category of sets. Clearly, the notion of symmetry is not univocal and strongly depends from the considered category. In the case of sets, the symmetry group is by definition simply the group of all permutations of the set's elements.

\begin{itemize}
 \item 
\emph{Let $\Gamma$ be a group and $C$ be the category of sets. Is there an object of $C$ such that (a subgroup of) its symmetry group $S_n$ is isomorphic to $\Gamma$?}
\end{itemize}

Yes, there is. Indeed, by Cayley's theorem every group $\Gamma$ is isomorphic to a subgroup of the symmetric group on $\Gamma$ (defined as the group of all bijections on $\Gamma$). 

Further examples include the following ones.

\begin{itemize}
 \item 
\emph{Let $\Gamma$ be a group and $C$ be the category of topological spaces. Is there an object $M$ of $C$ whose symmetry group (i.e., the group of all homeomorphisms on $M$) is isomorphic to $\Gamma$?}
\end{itemize}

Again, the answer is positive. Actually, such a topological space can be chosen to be a complete, connected, locally connected, 1-dimensional metric space: this has been proved by de Groot~\cite{Gro59}.

Symmetries also play an important r\^ole in graph theory. By definition, a {\bf symmetry} (or {\bf automorphism}) of a graph $G$ is a permutation of nodes of $G$ that preserves adjacency. Equivalently, a permutation is a symmetry of $G$ if and only if it commutes with the adjacency matrix of $G$. The group of all symmetries of a graph $G$ is usually denoted by $A(G)$. With this definition, the following can be formulated.

\begin{itemize}
 \item 
\emph{Let $\Gamma$ be a \emph{finite} group and $C$ the category of (simple) connected graphs. Is there an object $G$ of $C$ whose symmetry group $A(G)$ is isomorphic to $\Gamma$?}
\end{itemize}

Yes, there is. This affermative answer is the statement of Frucht's classical theorem~\cite{Fru38} (in fact, there exist infinitely many, pairwise non-isomorphic, finite graphs $G$ such that $A(G)\cong \Gamma$). This assertion has been significantly strengthened by a later work of Sabidussi~\cite{Sab57}, who has shown that these graphs can be constructed to be $k$-regular for any $k\ge 3$ and to have arbitrary connectivity and chromatic number.\footnote{Both Frucht and Sabidussi begin with the construction of a basic graph $G$ related to the Cayley graph of the group $\Gamma$ and then extend this construction to infinitely many further graphs by suitably decorating $G$. However, already the basic graph $G$ is in general highly redundant: e.g., the $3$-regular graph constructed in order to realize the symmetric group $S_n$ has 
$8\cdot n!$ nodes, whereas the Petersen graph' symmetry group is isomorphic to $S_5$ but the graph has only $10$ nodes.}
We also mention the following related results.
\begin{itemize}
 \item Let $\Gamma$ be a finite group. Then for any $k\in\{3,4,5\}$ there exist uncountably many, pairwise non-isomorphic, $k$-regular connected infinite  (simple) graphs $G$ such that $A(G)\cong \Gamma$~\cite{Izb59}. 
\item Let $\Gamma$ be an infinite group. Then there exists uncountably many, pairwise non-isomorphic (simple) connected infinite graphs $G$ such that $A(G)\cong \Gamma$~\cite{Gro59,Sab60}.
\end{itemize}

\section{Quantum graphs}

Let $H$ be a separable complex Hilbert space and $A$ a self-adjoint operator on $H$. By Stone's theorem, the abstract Cauchy problem of Schr\"odinger-type
\begin{equation}\label{sch}
iu'(t)=Au(t),\quad t\in\mathbb R,\qquad u(0)=u_0\in H,
\end{equation}
is well-posed and the solution $u$ is given by $u(t):=e^{itA}u_0$, where $(e^{itA})_{t\in\mathbb R}$ denotes the $C_0$-group of unitary operators on $H$ generated by $A$.

In the following we will consider closed operators $\Sigma:D(\Sigma)\to H$ such that
\begin{equation}\label{self}
\Sigma e^{itA}f=e^{itA}\Sigma f \qquad \hbox{for all }t\in\mathbb R\hbox{ and }f\in D(\Sigma)
\end{equation}

In theoretical physics, a unitary operator $\Sigma$ satisfying~\eqref{self} is said to be a {\bf symmetry} of the system described by~\eqref{sch}.

It has been observed in~\cite{CarMugNit08} that if $A$ is self-adjoint and dissipative (and hence it generates both a $C_0$-group $(e^{itA})_{t\in\mathbb R}$ of bounded linear unitary operators on $H$ and a $C_0$-semigroup $(e^{tA})_{t\ge 0}$ of linear contractive operators on $H$), then a closed subspace of $H$ is invariant under $(e^{itA})_{t\in\mathbb R}$ if and only if it is invariant under $(e^{tA})_{t\ge 0}$. Observe that self-adjoint dissipative operators are always associated with a (symmetric, $H$-elliptic, continuous) sesquilinear form, cf.~\cite{Ouh05}. 
The following criterion holds.

\begin{lemma}\label{ouha}
Let $a$ be a sesquilinear, symmetric, $H$-elliptic, continuous form with dense domain $D(a)$ associated with an operator $A$ on $H$. Consider a closed operator $\Sigma$ on $H$. 
Then $\Sigma$ satisfies~\eqref{self} if and only if
\begin{itemize}
 \item both $\Sigma ,\Sigma ^*$ leave $D(a)$ invariant and moreover
\item for all $f,g\in D(a)$
$$a(Lf+\Sigma ^*Rg,\Sigma ^*\Sigma Lf-\Sigma ^*Rg)=a(\Sigma Lf+\Sigma \Sigma ^*Rg,\Sigma Lf-Rg), $$
\end{itemize}
where $L:=(I+\Sigma ^*\Sigma )^{-1}$. $R:=(I+\Sigma \Sigma ^*)^{-1}$ and hence $I-R=\Sigma \Sigma ^*R$ and $I-L=\Sigma ^*\Sigma L$.
\end{lemma}

It follows immediately that if in particular $\Sigma$ is unitary, then it is a symmetry of the system described by~\eqref{sch} if and only if 
\begin{itemize}
 \item $\Sigma$ leaves $D(a)$ invariant and moreover
\item for all $f\in D(a)$
$$a(\Sigma f,\Sigma f)=a(f,f).$$
\end{itemize}

\begin{proof}
The proof of (1) is based on the observation that 
$$\Sigma e^{itA}=e^{itA}\Sigma \qquad\hbox{for all }t\in\mathbb R$$
if and only if the graph of $\Sigma$, i.e., the closed subspace
$${\rm Graph}(\Sigma):=\left\{\begin{pmatrix}x\\ \Sigma x\end{pmatrix}\in D(\Sigma)\times H\right\}$$
is invariant under the matrix group
$$\begin{pmatrix}
e^{itA} & 0\\ 0 & e^{itA}
\end{pmatrix},\qquad t\in\mathbb R,$$
on the Hilbert space $H\times H$, or equivalently under the matrix semigroup
$$\begin{pmatrix}
e^{tA} & 0\\ 0 & e^{tA}
\end{pmatrix},\qquad t\ge 0,$$
associated with the sesquilinear form ${\bf a}=a\oplus a$ with dense domain $D({\bf a}):=D(a)\times D(a)$.
A classical formula due to von Neumann yields that the orthogonal projection of $H\times H$ onto ${\rm Graph}(B)$ is given by
$$P_{{\rm Graph}(\Sigma)}=\begin{pmatrix}
(I+\Sigma ^*\Sigma )^{-1} & \Sigma ^*(I+\Sigma \Sigma ^*)^{-1}\\
\Sigma (I+\Sigma ^*\Sigma )^{-1} & I-(I+\Sigma \Sigma ^*)^{-1}
 \end{pmatrix}:=\begin{pmatrix}
L & \Sigma ^*R\\
\Sigma L & I-R
 \end{pmatrix},$$
cf.~\cite[Thm.~23]{Neu97}. The remainder of the proof is based on a known criterion by Ouhabaz, see~\cite[\S 2.1]{Ouh05}, stating that a closed subspace $Y$ of a Hilbert space is invariant under a semigroup associated with a form $b$ with domain $D(b)$ if and only if
\begin{itemize}
\item  the orthogonal projection $P_Y$ onto $Y$ leaves $D(b)$ invariant and 
\item $a(P_Y f,f-P_Y f)=0$ for all $f\in D(b)$. 
\end{itemize} 
Clearly 
$$P_{{\rm Graph}(\Sigma)} D({\bf a})\subset D({\bf a})$$ 
if and only if each of the four entries of $P_{{\rm Graph}(\Sigma)}$ leave $D(a)$ invariant. In particular, the upper-left entry leaves $D(a)$ invariant if and only if $\Sigma ^* \Sigma $ leaves $D(a)$ invariant, but then the lower-left entry leaves $D(a)$ invariant if and only if additionally $\Sigma $ leaves $D(a)$ invariant, too. Similarly, the lower-right entry leaves $D(a)$ invariant if and only if $\Sigma \Sigma ^*$ leaves $D(a)$ invariant, but then the upper-right (resp., lower-left) entry leaves $D(a)$ invariant if and only if additionally $\Sigma ^*$ (resp., $\Sigma$) leaves $D(a)$ invariant, too. Since however invariance of $D(a)$ under $\Sigma ,\Sigma ^*$ already implies invariance of $D(a)$ under $\Sigma \Sigma ^*,\Sigma ^*\Sigma $, the claim follows -- the second condition is in fact just a plain reformulation of Ouhabaz's second condition.
\end{proof}

A special class of Cauchy problems is given by so-called {quantum graphs}. In its easiest form (to which we restrict ourselves for the sake of simplicity), a {\bf quantum graph} $\mathcal G$ is a pair $(G,L)$, where $G=(V,E)$ is a (possibly infinite) simple connected graph and $L$ is an elliptic operator on $L^2(0,1)$. For technical reasons, edges have to be directed (in an arbitrary way which is not further relevant for the problem) and given a metric structure. Hence, we identify each edge $\overrightarrow{vw}$ with the interval $[0,1]$ and write $f(v):=f(0)$ and $f(w):=f(1)$ whenever we consider a function $f:\overrightarrow{vw}\equiv [0,1]\to {\mathbb C}$. 

To each quantum graph is naturally associated a system of Schr\"odinger type equations
$$i\frac{\partial \psi_e}{\partial t}(t,x)=\frac{\partial^2 \psi_e}{\partial x^2}(t,x),\qquad t\in\mathbb R,\; x\in (0,1),\; e\in E,$$
where $\psi_e:\overrightarrow{vw}\equiv e\equiv [0,1]\to {\mathbb C}$, i.e., $\psi$ are vector-valued wavefunctions from $[0,1]\to \ell^2(E)$. The natural operator theoretical setting of this problem includes the Hilbert space $H=L^2(0,1;\ell^2(E))\cong \prod_{e\in E}L^2(0,1;{\mathbb C})$ and the operator matrix defined by
$$\Delta \psi:={\rm diag}\left(\frac{\partial^2 \psi_e}{\partial x^2}\right)_{e\in E}.$$
Naturally, some compatibility conditions have to be satisfied in the boundary, i.e., in the nodes of the graph. These are given by
\begin{equation}\label{node1}
\psi_e(t,v)=\psi_f(t,v) \hbox{ for all } t\in\mathbb R,\; v\in V,\hbox{ whenever } e,f\sim v 
\end{equation}
(here and in the following we write $e\sim v$ if the edge $e$ is incident in the node $v$) and moreover
$$\sum_{e\sim v}\frac{\partial \psi_e}{\partial n}(t,v)=0\qquad t\in\mathbb R,\; v\in V.$$
Here $\frac{\partial \psi_e}{\partial n}$ denotes the normal derivative of $\psi_e$ at $0$ or $1$.

It can be easily shown that $\Delta$ is associated with the sesquilinear, symmetric, $H$-elliptic, continuous form $a$ defined by
$$a(\psi,\phi):=\int_0^1 (\psi'(x)|\phi'(x))_{\ell^2(E)}dx$$
with form domain
$$D(a):=\{\psi\in H^1(0,1;\ell^2(E)):\psi_e(v)=\psi_f(v)\qquad \hbox{for all }v\in V,\;\hbox{whenever } e,f\sim v 
\}.$$

Consistently with the general definition, a symmetry of a quantum graph $\mathcal G$ is a unitary operator on $H$ that commutes with the unitary group generated by $i\Delta$. Symmetries of $\mathcal G$ define a group, which we denote by ${\mathfrak A}({\mathcal G})$.

\section{Symmetries of quantum graphs}

\begin{itemize}
\item \emph{Let $\Gamma$ be a group and $C$ the category of quantum graphs. Is there an object $\mathcal G$ of $C$ whose symmetry group ${\mathfrak A}({\mathcal G})$ is isomorphic to $G$?}
\end{itemize}

The above question can be easily answered in the negative. Since $U(1)$ is always a subgroup of the symmetry group ${\mathfrak A}({\mathcal G})$ of a quantum graph, no finite group $\Gamma$ can be isomorphic to ${\mathfrak A}({\mathcal G})$. The above isomorphy condition can be relaxed, though.
 
%Moreover, one sees that~\eqref{node1} is equivalent to asking that 
%$$\begin{pmatrix}
%\psi_e(0)\\ \psi_e(1)
%\end{pmatrix}\in {\rm Range}\begin{pmatrix}
%({\mathcal I}^+)^T\\ ({\mathcal I}^-)^T  
% \end{pmatrix},$$
%where ${\mathcal I}^+, {\mathcal I}^-$ are the outgoing and incoming incidence matrices, respectively (cf.~\cite[\S2]{KraMugSik07}. Then it is easy to see that ...

In the proof of our main theorem we will need the notion of {\bf edge symmetry} of a graph: by definition, this is a permutation of edges of $G$ that preserves edge adjacency (or equivalently, a permutation of edges that commutes with the adjacency matrix of the line graph of $G$). In other words, by definition a permutation $\tilde{\pi}$ on $E$ is an edge symmetry if $\tilde{\pi}(e),\tilde{\pi}(f)$ have a common endpoint whenever $e,f\in E$ do. Edge symmetries forma group which is usually denoted by $A^*(G)$.

 Now, observe that each symmetry $\pi\in A(G)$ naturally induces an edge symmetry $\tilde{\pi}\in A^*(G)$: simply define
$$\tilde{\pi}(e):=(\pi(v),\pi(w))\qquad\hbox{whenever }e=(v,w).$$
While clearly $A'(G):=\{\tilde{\pi}:\pi \in A(G)\}$ (whose elements we call {\bf induced edge symmetries}) is a group, it can be strictly smaller than $A(G)$: simply think of the graph $G$ defined by
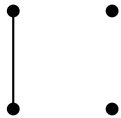
\begin{center}
\begin{tikzpicture}[style=thick]
\draw (0,0) -- (0,1.3);
\draw[fill] (0,0) circle (2pt);
\draw[fill] (0,1.3) circle (2pt);
\draw[fill] (1.3,0) circle (2pt);
\draw[fill] (1.3,1.3) circle (2pt);
\end{tikzpicture}
\end{center}
for which $A(G)=C_2\times C_2$ (independent switching of the adjacent nodes and/or of the isolated nodes) but $A'(G)$ is trivial. However, this is an exceptional case. The following has been proved by Sabidussi in the case of a finite group, but its proof (see~\cite[Thm. 1]{HarPal68}) carries over verbatim to the case of an infinite graph.

\begin{lemma}\label{harary}
Let $G$ be a simple graph. Then the groups $A(G)$ and $A'(G)$ are isomorphic provided that $G$ contains at most one isolated node and no isolated edge.
\end{lemma}

Hence, $A(G)\cong A'(G)$ in any connected graph with at least 3 nodes, and in particular in any connected, $3$-regular graph. 

\begin{theo}\label{main}
Let $\Gamma$ be a (possibly infinite) group. Then there exists a quantum graph $\mathcal G$ such that $\Gamma$ is isomorphic to a subgroup of ${\mathfrak A}({\mathcal G})$.
\end{theo}

\begin{proof}
To begin with, apply Frucht's theorem or its infinite generalization and consider some graph $G$ such that $A(G)$ is isomorphic to $\Gamma$: If $\Gamma$ is infinite consider an infinite connected graph yielded by the results of de Groot~\cite{Gro59} and Sabidussi~\cite{Sab60}, whereas if $\Gamma$ is finite consider a $3$-regular connected graph given by Frucht's theorem~\cite{Fru49}. In any case, $G$ is connected and has more than 3 nodes, hence by Lemma~\ref{harary} $A(G)\cong A'(G)$.

Now, any $\tilde{\pi}\in A'(G)$ can be associated with a bounded linear operator $\Pi$ on $H=L^2(0,1;\ell^2(E))$ defined by
\begin{equation}\label{induced}
 (\Pi f)_e:=f_{\tilde{\pi}(e)},\qquad e\in E.
\end{equation}
Such an operator is clearly unitary, since $\tilde{\pi}$ is a permutation, and the identifications
$$\pi\mapsto \tilde{\pi}\mapsto \Pi$$
define a group
$$\{\Pi\in {\mathcal L}(H):\pi\in A(G)\}\cong A(G)$$
of unitary operators on $H$. 

It remains to prove that each such $\Pi$ commutes with the unitary group $(e^{it\Delta})_{t\ge 0}$, or rather with its generator $\Delta$. In order to apply Lemma~\ref{ouha}, it suffices to observe that $\Pi$ is unitary and not dependent on the space variable, so that the second condition is trivally satisfied.

Finally, observe that if $\psi\in D(a)$, then clearly $\Pi \psi\in H^1(0,1;\ell^2(E))$ and moreover for all $v\in V$ and all $e,f\sim v$ one has
$$\Pi\psi_e(v)=\psi_{\tilde{\pi}(e)}(\pi(v))=\psi_{\tilde{\pi}(f)}(\pi(v))=\Pi\psi_f(v),$$
by definition of edge symmetry induced by $\pi$ and since by assumption
$$\psi_{e}(v)=\psi_{f}(v).$$
This yields invariance of $D(a)$ under $\Pi$ and concludes the proof.
\end{proof}

\begin{rem}
Clearly, any edge permutation $\tilde{\pi}$ induces a unitary operator $\Pi$ on $H$ defined as in~\eqref{induced}, but it generally ignores the adjacency structure of a graph. One could imagine that a general edge \emph{symmetry} $\tilde{\pi}\in A^*(G)$ may then suffice in the last part of the proof of Theorem~\ref{main}, in order to deduce that $\Pi$ is a symmetry of ${\mathfrak A}({\mathcal G})$. This is tempting, because on the one hand edge symmetries seem to be more general than induced edge symmetries (i.e., than elements of $A'(G)$), on the other hand they still preserve adjacency. However, general edge symmetries are not fit for our framework, as the following example shows: the permutation $\tilde{\pi}\equiv (e_1\; e_4)$ (which is not induced by any of the two (node) symmetries) is clearly an edge symmetry of the graph
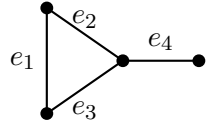
\begin{center}
\begin{tikzpicture}[style=thick]
\draw (0,0) -- (0,1.4) -- (1,0.7) -- cycle;
\draw (2,0.7) -- (1,0.7);
\draw[fill] (0,0) circle (2pt);
\draw[fill] (0,1.4) circle (2pt);
\draw[fill] (1,0.7) circle (2pt);
\draw[fill] (2,0.7) circle (2pt);
\draw[fill] (0,0.7) node[left]{$e_1$};
\draw[fill] (0.5,1) node[above]{$e_2$};
\draw[fill] (0.5,0.3) node[below]{$e_3$};
\draw[fill] (1.5,0.7) node[above]{$e_4$};
\end{tikzpicture}
\end{center}
but the induced unitary operator on $H$ does not preserve the  the boundary conditions~\eqref{node1}: in fact, both $e_1$ and $e_4$ are adjacent to (say) $e_2$, but their node which is common to $e_2$ is different.

Is there room for a generalization of Theorem~\ref{main} invoking edge permutations that are more general than those induced by (node) symmetries but less general than edge symmetries? In the absolute majority of cases the answer is negative: a classical result going back to Whitney (see~\cite[Cor.~9.5b]{BehChaLes79}) states that in a connected simple graph $G$ with at least three nodes the three groups $A(G),A'(G),A^*(G)$ are pairwise isomorphic if and only if $G$ is different from each of the following graphs:
\begin{center}
\begin{tikzpicture}[style=thick]
\draw (0,0) -- (0,1.4) -- (1,0.7) -- cycle;
\draw (2,0.7) -- (1,0.7);
\draw[fill] (0,0) circle (2pt);
\draw[fill] (0,1.4) circle (2pt);
\draw[fill] (1,0.7) circle (2pt);
\draw[fill] (2,0.7) circle (2pt);
\end{tikzpicture}
\qquad
\begin{tikzpicture}[style=thick]
\draw (0,0) -- (0,1.4) -- (1,0.7) -- cycle;
\draw (0,0) -- (2,0.7) -- (0,1.4);
\draw[fill] (0,0) circle (2pt);
\draw[fill] (0,1.4) circle (2pt);
\draw[fill] (1,0.7) circle (2pt);
\draw[fill] (2,0.7) circle (2pt);
\end{tikzpicture}
\qquad
\begin{tikzpicture}[style=thick]
\draw (0,0) -- (0,1.4) -- (1,0.7) -- cycle;
\draw (0,0) -- (2,0.7) -- (0,1.4);
\draw (1,0.7) -- (2,0.7);
\draw[fill] (0,0) circle (2pt);
\draw[fill] (0,1.4) circle (2pt);
\draw[fill] (1,0.7) circle (2pt);
\draw[fill] (2,0.7) circle (2pt);
\end{tikzpicture}
\end{center}
\end{rem}

\bibliographystyle{plain} 
\bibliography{/home/delio/Dropbox/mate/literatur.bib}
\end{document}